\documentclass{article}
\usepackage{authblk}
\usepackage[english]{babel}
\usepackage{complexity}
\usepackage[letterpaper,top=2cm,bottom=2cm,left=3cm,right=3cm,marginparwidth=1.75cm]{geometry}

\usepackage{amsmath,amsthm}
\usepackage{graphicx}
\usepackage{amsfonts}
\usepackage[colorlinks=true, allcolors=blue]{hyperref}

\newtheorem{theorem}{Theorem}
\newtheorem{lemma}[theorem]{Lemma}
\newtheorem{corollary}[theorem]{Corollary}
\usepackage{tikz}
\usetikzlibrary{quantikz}

\title{Quantum Algorithm for Estimating  Gibbs Free Energy and Entropy via Energy Derivatives}
\author[1]{Shangjie Guo}
\author[1]{Corneliu Buda}
\author[2,3,4]{Nathan Wiebe}
\affil[1]{Innovation and Digital Science, bp, Houston TX, USA}
\affil[2]{University of Toronto, Toronto ON, Canada}
\affil[3]{Pacific Northwest National Laboratory, Richland WA, USA}
\affil[4]{Canadian Institute for Advanced Research, Toronto ON, Canada}
\begin{document}
\maketitle

\begin{abstract}
Estimating vibrational entropy is a significant challenge in thermodynamics and statistical mechanics due to its reliance on quantum mechanical properties. This paper introduces a quantum algorithm designed to estimate vibrational entropy via energy derivatives. 
Our approach block encodes the exact expression for the second derivative of the energy and uses quantum linear systems algorithms to deal with the reciprocal powers of the gaps that appear in the expression.  We further show that if prior knowledge about the values of the second derivative is used then our algorithm can $\epsilon$-approximate the entropy using a number of queries that scales with the condition number $\kappa$, the temperature $T$, error tolerance $\epsilon$ and an analogue of the partition function $\mathcal{Z}$, as $\widetilde{O}\left(\frac{\mathcal{Z}\kappa^2 }{\epsilon T}\right)$.  We show that if sufficient prior knowledge is given about the second derivative then the query scales quadratically better than these results.  This shows that, under reasonable assumptions of the temperature and a quantum computer can be used to compute the vibrational contributions to the entropy faster than analogous classical algorithms would be capable of.
Our findings highlight the potential of quantum algorithms to enhance the prediction of thermodynamic properties, paving the way for advancements in fields such as material science, molecular biology, and chemical engineering.
\end{abstract}

\section{Introduction}
Recent advancements in quantum computing have led to the development of a host algorithms for estimating ground state energies~\cite{reiher2017elucidating,lin2025dissipative,tang2021qubit,lee2021even}. These algorithms have shown promise in providing accurate energy estimates for complex quantum systems, opening new avenues for studying molecular and material properties at a quantum level~\cite{o2022efficient}.  One of the most significant applications considered for quantum computing is in catalysis wherein the reaction constants between two different configurations of a molecular system can be estimated using absolute rate theory to scale as $e^{-\Delta G/k_bT}$ where $\Delta G$ is the Gibbs free energy difference between two configurations  $k_b$ is Boltzmann's constant and $T$ is the temperature of the system.

Despite the significance of this problem, relatively little work has been done on estimating the free energy differences for quantum systems.  Instead, most work pushes towards this aim by studying the problem of probing the groundstate energy in different configurations.  This is a highly desirable quantity because $\Delta G = \Delta H - T\Delta S$, where $\Delta H$ is the enthalpy difference and $\Delta S$ is the entropy difference between two configurations and the groundstate energy difference provides corresponds to the enthalpy difference between two  configurations~\cite{rosenstock1952absolute}.  The entropic contribution, however, is much more subtle to compute.  Exact quantum algorithms for estimating the entropy scale polynomially with the Hilbert space dimension optimally~\cite{acharya2020estimating} and generically the problem of preparing the thermal distributions corresponding to low temperature is a hard computational problem using all known algorithms~\cite{chen2023quantum}.  Approximate methods exist for estimating the entropy~\cite{simon2024improved,huang2025fullqubit} however these methods can require uncontrollable approximations and can scale poorly in some cases.




The primary objective of this article is to establish a new algorithm for estimating the Gibbs free energy through the derivatives of the groundstate energy using quantum computing techniques.  
By accurately calculating the vibrational contributions to entropy, we aim to enhance the understanding and prediction of thermodynamic properties in molecular systems, thereby contributing to advancements in various scientific and industrial fields.  Specifically, the need for accurate vibrational entropy estimations becomes critical in fields such as material science, molecular biology, and chemical engineering.  We further show that our algorithm for estimating the vibrational energy is not dequantizable under reasonable complexity theoretic assumptions because the problem is $\BQP$-complete.  This means that both estimating the energy as well as the free energy is a problem that is beyond the scope of classical computers and further problematizes the previous perspective that entropy differences can be accurately approximated in polynomial time on a classical computer~\cite{reiher2017elucidating} without making further restrictions.

\section{Entropy Approximation}
We follow a standard approach used by computational chemists to estimate the entropy for a chemical system.  Our approach to entropy computation breaks the entropy into a series of different contributions.  These contributions arise from the vibrational degrees of freedom, the translational degrees of freedom, the rotational degrees of freedom and the electronic contribution.  Here we work under the Born-Oppenheimer approximation, which allows us to consider the electronic degrees of freedom separately from the nuclear degrees of freedom.  Under this approximation, we can write the entropy as
\begin{equation}
    S=S_{\rm vib}+S_{\rm trans}+S_{\rm rot}+S_{\rm el}.
\end{equation}
Of these four pieces, the vibrational contribution is the most involved.  The remaining pieces do not require quantum resources to estimate them under these standard approximations~\cite{mcquarrie2000statistical}.  For this reason, we discuss the remaining contributions in turn.

We assume in the following that the pressure that the system experiences is fixed throughout a process as is the temperature.  We further neglect intermolecular forces (assume the ideal gas equation), ignore solvent effects, assume the gap between the groundstate and the first excited state is infinite, the temperature is large compared to the rotational energy scale and that the groundstate itself is promised to be non-degenerate.  These approximations are almost never precisely held in physical settings; however, such approximations are often appropriate for physical systems~\cite{mcquarrie2000statistical}. Despite this,  precise bounding of the errors that arise from approximating them remains an interesting avenue of inquiry for further work that aims to understand the exact cost of free energy computations.

The translational contribution to the entropy can be computed in the following manner.  It involves computing the partition function for this degree of freedom.  If we treat the individual molecules in the simulation as an ideal gas, then we can use the ideal gas equation to estimation the partition function for the motion of the center of of the system under these approximations.  This approximation leads to~\cite{mcquarrie2000statistical}
\begin{equation}
    S_{\rm trans} = k_B\left(\ln\left(\left(\frac{2\pi m k_B T}{h^2} \right)^{3/2}\frac{k_BT}{P}\right)+\frac{5}{2}\right).
\end{equation}
Here $m$ is the total mass of the system, $T$ is the temperature, $P$ is the external pressure and $h$ is Planck's constant.
Under the above assumptions all these terms can be computed analytically.

The rotational contribution to the entropy is similarly given by the rotational partition function.  Under the assumption that the temperature is large relative to the energy scale of the vibrational motion~\cite{mcquarrie2000statistical}, the expression for the rotational entropy is
\begin{equation}
    S_{\rm rot}= k_B\left( \ln\left(\frac{\sqrt{\pi}T^{3/2}}{\sigma_r\sqrt{\Theta_{x}\Theta_y \Theta_z}} \right)+\frac{3}{2}\right),
\end{equation}
where $\sigma_r$ is the rotational symmetry number which gives the number of ways that the system can be rotated back into itself and $\Theta_x,\Theta_y,\Theta_z$ are constants related to the moments of inertia for the molecular system about all three principal axes.

The electronic entropy is trivial under our assumptions.  This is because if the temperature is finite and the spectral gap is infinite then the electronic thermal state will be
\begin{equation}
    \rho = e^{-H/k_BT}/\mathcal{Z} = \ket{\psi_0}\!\bra{\psi_0}
\end{equation}
if the groundstate is non-degenerate.  The entropy of this state is zero and hence
\begin{equation}
    S_{\rm el} = 0.
\end{equation}

Thus the entropy under these approximations can be expressed as
\begin{equation}
S = S_{\rm vib} +k_B\left(\ln\left(\left(\frac{2\pi m k_B T^2}{h^2} \right)^{3/2}\left(\frac{\sqrt{\pi}k_BT}{\sigma_rP\sqrt{\Theta_{x}\Theta_y \Theta_z}} \right) \right)+4 \right).
\end{equation}
The vibrational entropy remains, as yet, unspecified.  This is because it has a character that differs from the other contributions as vibrations are not related to the rigid motion of a molecule but rather is due to deformations of the system and can in some cases be thought of as a phonon bath that the electronic degrees of freedom couple to.  This makes the vibrational degrees of freedom quantities that need further input from the quantum mechanical model to compute as we need to probe the degree of resistance to stretching, rocking or other related modes of motion that the electrons provide.  This necessitates a quantum mechanical model for such a contribution and computation of $S_{\rm vib}$ is our main focus for this work.

\section{Vibrational Entropy as a Function of Energy Derivative}

In this section, we focus on deriving the vibrational entropy of a system as a function of energy derivatives. 
Vibrational entropy is a critical component in determining the thermodynamic properties of molecular systems, especially at high temperatures. 
We present a method to estimate vibrational entropy using quantum mechanical principles and the second derivative of the system's energy.

Let \( H \) be the Hamiltonian of a quantum system at equilibrium. Consider a small imaginary displacement \( \delta_i \) from the equilibrium value of the degree of freedom \( x_i \) that contributes to the entropy of the system. We have \( D \) degrees of freedom contributing to the entropy, indexed by \( i \) such that \( 0 \leq i \leq D-1 \).

The perturbed Hamiltonian can be expressed as:
\begin{equation}
\tilde{H} = H + \delta_i V_i
\end{equation}
where \( V_i \) is the Hermitian operator corresponding to the change \( \delta_i \). The spring constant \( k_i \) for the degree of freedom \( x_i \) can be derived from the second derivative of the system's energy:
\begin{equation}
k_i = \frac{\partial^2 E_0}{\partial x_i^2}
\end{equation}
where \( E_0 \) is the ground state energy of the Hamiltonian \( H \).

Using the spring constants \( k_i \), we can derive the characteristic temperature \( \theta_i \):
\begin{equation}
\theta_i = \frac{\hbar}{k_B} \sqrt{\frac{k_i}{\mu_i}}
\end{equation}
where \( \hbar \) is the reduced Planck constant, \( k_B \) is the Boltzmann constant, and \( \mu_i \) is the reduced mass for the degree of freedom \( x_i \).

The total vibrational entropy \( S_{\text{vib}} \) can be expressed as~\cite{mcquarrie2000statistical}:
\begin{equation}
S_{\text{vib}} = k_B \sum_{i} \left( \frac{\theta_i/T}{e^{\theta_i/T} - 1} - \ln \left( 1 - e^{-\theta_i/T} \right) \right)
\end{equation}
where \( T \) is the temperature.

By accurately calculating the vibrational contributions to entropy, we aim to enhance the understanding and prediction of thermodynamic properties in molecular systems, thereby contributing to advancements in various scientific and industrial fields.
 Our aim is then to provide a quantum algorithm that can be used to estimate the vibrational entropy as given above.  However, the first thing that we need for this is a way of computing the second derivative of the energy along each direction.

\begin{lemma}[Second Derivative Lemma]
Let $H(s) = H_0 + sV$ where $\bra{E_k} V \ket{E_j} = 0$ if $E_j = E_k$ and $j\ne k$ for $V$  a bounded operator acting on the same Hilbert space and $K = V^\dagger (E_0 I - H_0)^{+} V$.
 then we have that for any eigenstate of $H_0$, $\ket{E_0}$, with known eigenvalue $E_0$ we have that
$$
\left.\partial_s^2 E_0(s) \right|_0 = \bra{E_0} K \ket{E_0}.
$$
\end{lemma}
\begin{proof}
From perturbation theory, we have that under the assumption of non-degeneracy for specified degree of freedom $\delta_i = \delta$ the new  eigenenergy of \( H \) as a function of \( \delta \) can be expanded as:
\begin{equation}
E_0 (\delta) = E_0 + \delta \langle E_0|V|E_0 \rangle + \delta^2 \sum_{k \neq 0} \frac{|\langle E_k|V|E_0 \rangle|^2}{E_0 - E_k} + \mathcal{O}(\delta^3)
\end{equation}
Then from Taylor's theorem as the Energy is a three-times differentiable function under the assumption that $H_0$ is non-degenerate we have that
\begin{equation}
    \lim_{\delta \rightarrow 0}\left.\partial_\delta^2 E_0(\delta)\right|_0 = \lim_{\delta \rightarrow 0}\left(2 \sum_{k \neq 0} \frac{|\langle E_k|V\ket{E_0}|^2}{E_0 - E_k} + \mathcal{O}(\delta)\right)=2 \sum_{k \neq 0} \frac{|\langle E_k|V|E_0 \rangle|^2}{E_0 - E_k}.\label{eq:2deriv}
\end{equation}


The goal is to find a block encoding for the operator \( K \) such that:
\begin{equation}
\langle E_0 |K|E_0 \rangle = 2\sum_{k \neq 0} \frac{|\langle E_k|V|E_0 \rangle|^2}{E_0 - E_k}
\end{equation}
where \( H = H_0 + \lambda V \) and \( H_0 |k \rangle = E_k |k \rangle \).  This operator is not trivial to express because of the fact that it depends on the eigenvalues of $H_0$, which are not provided explicitly. 
We tackle this problem by instead using an expression for the inverse function. 

From the relationship:
\begin{equation}
H_0^{-1} |k \rangle = \frac{1}{E_k} |k \rangle \quad \text{since} \quad |k \rangle = E_k H_0^{-1} |k \rangle
\end{equation}
we can rewrite the sum over states \( k \) as:
\begin{equation}
\sum_{k \neq n} \frac{|\langle k|V|n \rangle|^2}{E_n - E_k} = \sum_{k} \frac{|\langle k|V|n \rangle|^2}{E_n - E_k} = \sum_{k} \frac{\langle n|V^\dagger|k \rangle \langle k|V|n \rangle}{E_n - E_k} = \langle n | V^\dagger \sum_{k} \frac{|k \rangle \langle k|}{E_n - E_k} V |n \rangle = \langle n |K|n \rangle
\end{equation}
This expression shows how the operator \( K \) can be written in terms of the perturbation \( V \) and the energy differences.

Therefore, we have:
\begin{equation}
K = V^\dagger \sum_{k} \frac{|k \rangle \langle k|}{E_n - E_k} V
\end{equation}
Assuming \( E_n \) is a known constant, we can write:
\begin{equation}
\frac{1}{{E_n - E_k}} = \langle k | (E_n I - H_0)^{+} |k \rangle \quad \text{and since} \quad H_0 |k \rangle \langle k| = E_k |k \rangle \langle k| = |k \rangle \langle k| E_k = |k \rangle \langle k| H_0
\end{equation}
we get:
\begin{equation}
K = V^\dagger \sum_{k} |k \rangle \langle k | (E_n I - H_0)^{+} |k \rangle \langle k| V = V^\dagger \sum_{k} (|k \rangle \langle k |)^2 (E_n I - H_0)^{+} V = V^\dagger (E_n I - H_0)^{+} V
\end{equation}
Since \( V \) is a Hamiltonian and therefore Hermitian, we let \( V^\dagger = V \) for simplicity. Here, \( A^+ \) is the Moore–Penrose inverse of \( A \).
\end{proof}

\section{Energy Derivative as Block Encoding}

In this section, we present the methodology to encode the second derivative of the ground state energy as a block encoding. This involves using perturbation theory and quantum algorithms to evaluate the energy derivatives and construct the required block encoding operators. This is a crucial step in estimating the vibrational entropy of a system.

\begin{theorem}\label{thm:blockEnc}
Under the assumptions that the matrix $V =\sum_{j=1}^L b_j U_j$ and $E_0I - H_0 =\sum_{j=0}^J a_j U_j$ for a family of unitary operators $U_\ell$ that are provided by block encodings using unitary prepare circuits $P_V,P_{H'}$ and select operations, $U_V,U_{H'}$, a unitary $U_K$ can be constructed that is a $(O(|b|^2 \kappa, \log(LJ), \epsilon)$ block encoding of $K$ that uses a number of queries to $P_V,P_{H'}$ and the select operations $U_V,U_{H'}$ 
$$
N_{\rm queries}(U_K) \in O\left( |b|^2 \kappa \log(\kappa |b|^2/\epsilon) \right)
$$
\end{theorem}
The proof of this result involves a number of technical steps which require the following lemmas.
The first such lemma involves showing that a block encoding \( U_K \) of \( K \) can be constructed by the multiplication of three block encodings which can be implemented using the techniques of \cite{gilyen2019quantum}.  We illustrate this process in Figure~\ref{fig:multBlock} and stated in the following Lemma.

\begin{figure}[t]
\begin{center}
\begin{quantikz}
&\lstick{$|{0}\rangle^{\otimes a_L}$} & \gate{P_V} & \gate{U_V}\arrow[d,  no head, no tail] & \qw & \qw 
& \gate{P_V^\dagger} & \meter{}
\\
&\lstick{$|{0}\rangle^{\otimes a_A}$} & \gate{P_A} & \qw\arrow[d,  no head, no tail] & \gate{U_A}\arrow[d,  no head, no tail] & \qw & \gate{P_A^\dagger} & \meter{}
\\
&\lstick{$|{0}\rangle^{\otimes a_L}$} & \gate{P_V} & \qw\arrow[d,  no head, no tail] & \qw\arrow[d,  no head, no tail] & \gate[2]{U_V} & \gate{P_V^\dagger} & \meter{}
\\
&\lstick{$|{\psi}\rangle$} & \qw & \gate[label style={white}]{U_V} & \gate[label style={white}]{U_A} & \qw & \qw & \qw
\end{quantikz}
\end{center}
\caption{Recursive block encodings used in the block-encoding of $VAV$. \label{fig:multBlock}}
\end{figure}
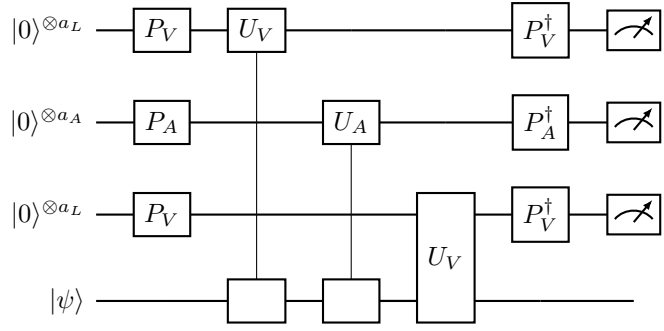

\begin{lemma}\label{lem:multBlock}
    Let $F_V, F_V',F_A$ be unitary operations on a Hilbert space $\mathcal{H} = \mathcal{H}_V \otimes \mathcal{H}_{A}\otimes \mathcal{H}_{V'}\otimes \mathcal{H}_{sys}$ such that for constants $\lambda_V, \lambda_A$
    \begin{eqnarray}
        (\bra{0}\otimes I) U_V (\ket{0}\otimes I) &=& I_A\otimes I_{V'}\otimes V/\lambda_V\nonumber\\
        (I\otimes \bra{0}\otimes I) U_A (I\otimes\ket{0}\otimes I) &=& I_V\otimes I_{V'}\otimes A/\lambda_A\nonumber\\
        (I\otimes \bra{0}) U_{V'} (I\otimes \ket{0}) &=& I_V\otimes I_{A}\otimes V/\lambda_V
    \end{eqnarray}
    Under these circumstances, the circuit of Fig.~\ref{fig:multBlock} provides a $(\lambda_V^2\lambda_A,2a_L+a_A,0)$-Block Encoding of $VAV$.
\end{lemma}
\begin{proof}
    From Lemma 30 of~\cite{gilyen2019quantum} we have that if $U_a$ and $U_b$ are $(\alpha,n_a,\epsilon)$ and $(\beta,n_b,\delta)$ block encodings of matrices $a,b$ respectively that act on the same domain we have that $(I_a\otimes U_b)(I_b\otimes U_a)$ is an $(\alpha\beta,n_a+n_b,\beta\epsilon+\alpha\delta)$ block encoding of $ab$.  Now consider $U_c$ an $(\gamma,n_a+n_b, \eta)$ block encoding of a matrix $c$ that acts on the same domain as $a,b$.  We can then apply the same result to see that $(I_a\otimes I_b \otimes U_c)(I_a\otimes I_c\otimes U_b)(I_b\otimes I_c\otimes U_a)$ is an $(\alpha\beta \gamma,n_a+n_b+n_c,\beta\gamma\epsilon+\alpha\gamma\delta +\alpha\beta\eta)$ bloxk encoding of $abc$.  The circuit in Figure~\ref{fig:multBlock}
 is a circuit representation of that product. 
 Thus because the three unitaries in question block encode $V,A$ with constants $\alpha=\lambda_V$, $\beta = \lambda_A$, $\gamma =\lambda_V$ the overall block encoding constant is $\lambda_V^2 \lambda_A$.  Given that the error in all three block encodings is zero and $n_b=\alpha_A$ and $n_a=n_c=\alpha_L$ the claimed memory and error for the block encoding follows. \end{proof}


We use the Linear Combination of Unitaries (LCU) method to construct the block encoding for the operator \( V \). For physical systems like molecules, \( V \) is often a Hermitian operator that can be decomposed into a summation of Pauli operators. Specifically, let:
\begin{equation}
V = \sum_{l=1}^L b_l U_l
\end{equation}
be a linear combination of unitaries, where \( U_l \) are unitary operators and \( b_l \) are real numbers such that \( b_l > 0 \). The sum of the coefficients is denoted by \( |b| = \sum_{l=1}^L b_l \).

We define the preparation operator \( P_{V} \) as:
\begin{equation}
P_{V}: |0\rangle \rightarrow \sum_{l=1}^L \sqrt{\frac{b_l}{|b|}}|l\rangle
\end{equation}
which prepares a quantum state corresponding to the coefficients \( b_l \). The unitary operation \( U_V \) is defined as:
\begin{equation}
U_V: |l\rangle|\psi\rangle \rightarrow |l\rangle U_l|\psi\rangle
\end{equation}
which applies the unitary \( U_l \) conditioned on the state \( |l\rangle \).

Applying the LCU method, we achieve the block encoding of \( V \). The circuit for this block encoding is given by:

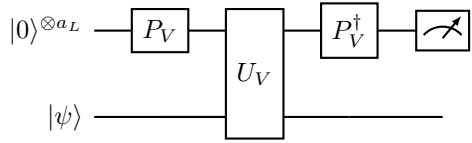
\begin{center}
\begin{quantikz}
&\lstick{$|{0}\rangle^{\otimes a_L}$} & \gate{P_V} & \gate[2]{U_V} & \gate{P_V^\dagger} & \meter{}
\\
&\lstick{$|{\psi}\rangle$} & \qw & \qw & \qw & \qw
\end{quantikz}
\end{center}

This circuit represents a \((|b|, a_L, 0)\) block encoding of \( V \) with success probability:
\begin{equation}
\frac{\langle \psi | V^\dagger V | \psi \rangle}{|b|^2} = \frac{\langle \psi | V^2 | \psi \rangle}{|b|^2}
\end{equation}
where \( a_L = \lceil \log_2 (L) \rceil \). The success probability indicates the likelihood of measuring the desired outcome and correctly encoding the operator \( V \). This block encoding allows us to represent the operator \( V \) as a quantum circuit that can be efficiently implemented and used in further quantum computations.

We summarize these observations in the following Lemma.
\begin{lemma}\label{lem:VCost}
    Let $V = \sum_{l=1}^L b_l U_l$ for a set of unitary operations $U_l$ that are accessed through the prepare unitary $P_V$ and the select unitary $U_V$ defined above.  A $(|b|,\lceil L \rceil, 0)$ block-encoding of $V$ can be prepared using $2$ queries to $P_V$ or its inverse and one query to $U_V$.
\end{lemma}

The pseudo-inverse of the shifted Hamiltonian can be implemented using the quantum algorithms~\cite{wiebe2012quantum, childs2017quantum,gilyen2019quantum}.  Specifically, if we are given a block encoding of $H' =E_0 I - H_0 = \sum_{j=0}^J a_j U_j$ for a set of unitary operators $U_j$ for positive $a_j$.  Specifically, the idea behind this approach is to use the quantum singular value transformation to convert a block encoding of $H'$ into that of $(E_0 I - H_0)^+$.  Specifically, let us assume that the spectrum of $H'/|a| = H'/(|E_0| + \sum_j a_j)$ is confined to the interval $[-1, -1/\kappa)\cup (1/\kappa,1]$ for $\kappa\ge 1$.  We then have that from Theorem 41 of~\cite{gilyen2019quantum}
\begin{theorem}\label{thm:pseudo}
    Let $U_{H'}$ be an $(|a|, \lceil \log(J)\rceil,0)$ block-encoding of $H'$ subject to the above assumptions on the eigenvalues of $H'$.  We then have that an $(O(\kappa, \log(J), \epsilon)$ block encoding of $((E_0 I + H_0))^{+}$ can be constructed using $O(\kappa \log(1/\epsilon))$ applications of $U_{H'}$ where $\kappa$ is an upper bound on the condition number of $H'$.
\end{theorem}


\begin{proof}[Proof of Theorem~\ref{thm:blockEnc}]
Our aim is now to block encode $K= V^\dagger A V$.  We block encode $K$ by using Lemma~\ref{lem:multBlock} to block encode the product of the operators.  The cost of block-encoding $V$ or its adjoint is given in Lemma~\ref{lem:VCost} and the cost of block-encoding $A$ is given in Theorem~\ref{thm:pseudo}.  The normalization constant is at most the product of the normalization constants of the constituent matrices, leading to a normalization constant of $O(|b|^2\kappa)$.  The number of qubits needed is on the order of the sum of that needed for each of the block-encoded matrices.  This is $O(\log(JL))$.  The error is only due to that of the pseudo-inverse.  This error is multiplied through by $|b|^2\kappa$ and thus in order to compensate for this we must comenserately shrink the error in the block-encoding to $O(\epsilon /\kappa|b^2|)$ which leads to $O(\kappa \log(|b|^2\kappa/\epsilon))$ applications of $U_{H'}$ needed for the decomposition thus the result follows. \end{proof}

\section{Evaluation of operator $K$}
Our aim here is to use the Hadamard test to provide a method to compute the expectation value of the second derivative operator in a given state.  We do this by noting that the expectation value of the block can be given by the Hadamard test post selected on the block being applied. The Hadamard test circuit takes the following form:

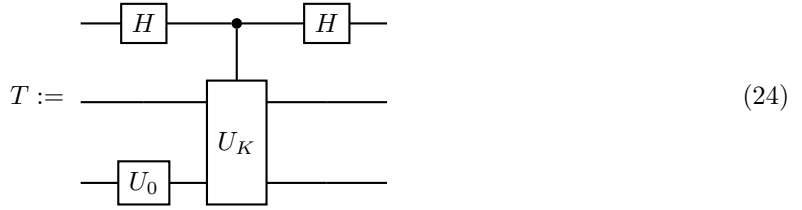
\begin{equation}
T:=\begin{quantikz}
& \gate{H} & \ctrl{1} & \gate{H} & \qw  \\
 & \qw & \gate[wires=2]{U_K} & \qw & \qw \\
 & \gate{U_0} & & \qw & \qw 
\end{quantikz}    
\end{equation}

Define the ``good state'' projection operator:
\begin{equation}
P_G = |0\rangle \langle 0| \otimes I^{\otimes (m+n)}
\end{equation}

Then we define the success probability to be for block encoding constant $\alpha$ of $K$s
\begin{equation}
P(0) := \|P_G T |0\rangle^{\otimes (m+n+1)} \|^2 = {\frac{1}{2}(1+\frac{1}{\alpha}\operatorname{Re} \langle E_0 | K | E_0 \rangle)}
\end{equation}
$K$ is also a Hermitian operator, and as such the probability of measuring $0$ in the Hadamard test can be expressed as 
\begin{equation}
P(0)=  {\frac{1}{2}(1+\frac{1}{\alpha} \langle E_0 | K | E_0 \rangle)} \label{eq:P0}
\end{equation}
Amplitude estimation can then be used to  estimate the probability $P(0)$ via the following lemma which is a trivial modification of the results of Theorem 12 of~\cite{brassard2000quantum}.
\begin{lemma}[Amplitude Estimation]\label{lem:AE}
    There exists a quantum algorithm that yields an estimate $\hat{P}(0)$  with probability at least $8/\pi^2$  using $M$ queries to  $U_0$ and Controlled-$U_K$ such that
    $$
    |\hat{P}(0) - P(0)|\le 2\pi \frac{\sqrt{P_{\max}(1-P_{\min})}}{M}+\frac{\pi^2}{M^2}
    $$
    Such that $P(0) \in [P_{\min},P_{\max}]\subseteq [0,1]$.
\end{lemma}

Thus the expectation value can be extracted with high probability from the amplitude estimation process, which can be carried out using the application of the walk operator
 \begin{equation}
 W = - (I - 2 T |0\rangle \langle 0| T^\dagger) (I - 2 P_G)
\end{equation}
as illustrated in Figure~\ref{fig:QAE}.  We then use this to estimate the overlap and achieve better-than-Heisenberg scaling in cases where we have accurate estimates of the expectation value of $K$ by borrowing insights from~\cite{simon2024amplified,king2025quantum}.  This result is stated below.

\begin{lemma}\label{lem:KBE}
    Assume that constants $\langle K\rangle_{\min}$ and $\langle K\rangle_{\max}$ are known such that $\bra{E_0} K \ket{E_0} \in [\langle K\rangle_{\min}, \langle K\rangle_{\max}]$.  For any $\epsilon>0$ and $\delta>0$ there exists a quantum algorithm that can produce an estimate $\langle \hat{K} \rangle$ such that with probability at least $1-\delta$ we have that the error in the estimate obeys
    $| \langle \hat{K}\rangle -\bra{E_0} K \ket{E_0}|\le \epsilon
    $
    that queries $U_0$ and $U_K$ a number of times that is at most
    $$
N_{\rm queries} \in O\left(\frac{\alpha\log(1/\delta)}{\epsilon}\left(\sqrt{\left( \frac{1}{2} + \frac{\langle K_{\max}\rangle }{2\alpha}\right)\left(\frac{1}{2} - \frac{\langle K_{\min}\rangle }{2\alpha} \right)}  +\sqrt{\epsilon/\alpha} \right) \right)
    $$
\end{lemma} 
\begin{proof}
From Lemma~\ref{lem:AE} we have that if we wish to obtain error $\epsilon_0$ in our estimate then it suffices to choose
\begin{equation}
    \pi^2 + 2\pi M \sqrt{P_{\max}(1-P_{\min})} -\epsilon_0 M^2 =0.
\end{equation}
The solution to this quadratic equation is
\begin{equation}
    M\ge \frac{\pi\left( \sqrt{ P_{\max}(1-P_{\min})} + \sqrt{ P_{\max}(1-P_{\min}) + \epsilon_0}\right)}{\epsilon_0} 
\end{equation}
and this condition is further implied by (using the expression $\sqrt{1+x} \le 1+x/2$),
\begin{align}
    M&\ge \frac{\pi(6P_{\max}(1-P_{\min})\sqrt{\epsilon_0}+(P_{\max}(1-P_{\min}))^{3/2}+\epsilon_0^{3/2}+2\sqrt{P_{\max}(1-P_{\min})}\epsilon_0)}{4\sqrt{P_{\max}(1-P_{\min})}\epsilon_0^{3/2}}\nonumber\\
    &\in O\left(\frac{\sqrt{P_{\max}(1-P_{\min})}}{\epsilon_0} + \frac{1}{\sqrt{\epsilon_0}} \right)
\end{align}
Next we have from Lemma~\ref{lem:AE} that the probability of success is at least $8/\pi^2$.  Let us imagine that we repeat the experiment $D$ times and take the median of the result.  The probability that out of the $D$ samples that more than $D/2$ are correct is given by
\begin{equation}
    P_{\rm good } \le 1-e^{-D(8-\pi^2/2)/16 }.
\end{equation}
If the number of samples that are correct is greater than $D/2$ then the median estimate is typical of the good set and thus will have error at most $\epsilon_0$.  Thus we can obtain a maximum probability of failure of $\delta$ by choosing $D$ such that 
\begin{equation}
    D\in O(\log(1/\delta)).
\end{equation}
As the entire algorithm is repeated $D$ times the cost of the algorithm is $D N_{\rm queries}$ which evaluates to 
\begin{equation}
    O\left(\frac{\log(1/\delta)}{\epsilon_0}\left(\sqrt{P_{\max}(1-P_{\min})}  +\sqrt{\epsilon_0} \right) \right)
\end{equation}

Next from~\eqref{eq:P0} we have that the probability is 
\begin{equation}
    \frac{1}{2} + \frac{\langle K_{\max}\rangle }{2\alpha}\ge\frac{1}{2}(1+\frac{1}{\alpha} \langle E_0 | K | E_0 \rangle) \ge  \frac{1}{2} + \frac{\langle K_{\min}\rangle }{2\alpha}.
\end{equation}
Thus we have that we can take the upper and lower bounds on the success probability to be the above bounds and in turn the number of queries needed reads
\begin{equation}
    O\left(\frac{\log(1/\delta)}{\epsilon_0}\left(\sqrt{\left( \frac{1}{2} + \frac{\langle K_{\max}\rangle }{2\alpha}\right)\left(\frac{1}{2} - \frac{\langle K_{\min}\rangle }{2\alpha} \right)}  +\sqrt{\epsilon_0} \right) \right)
\end{equation}

Finally the error that we wish to obtain in our final estimate of the derivative is $\epsilon$.  From~\eqref{eq:P0} we have that we can construct an estimate of $\bra{E_0} K \ket{E_0}$, $\langle \hat{K}\rangle$  via 
\begin{equation}
    \langle K\rangle = \alpha(2\hat{P}(0)-1)
\end{equation}
Thus in order to ensure that the error in $\langle K \rangle $ is at most $\epsilon,$ it suffices to choose $\epsilon_0 \in \Theta(\epsilon/\alpha)$ thus the overall number of queries is
\begin{equation}
    O\left(\frac{\alpha\log(1/\delta)}{\epsilon}\left(\sqrt{\left( \frac{1}{2} + \frac{\langle K_{\max}\rangle }{2\alpha}\right)\left(\frac{1}{2} - \frac{\langle K_{\min}\rangle }{2\alpha} \right)}  +\sqrt{\epsilon/\alpha} \right) \right)
\end{equation}
as claimed.
\end{proof}

\begin{corollary}\label{cor:FreeE}
Let $V=\sum_j b_j U_j$ and $E_0I-H_0=\sum_j a_j U_j$ accessed by prepare oracles $P_V,U_V$ and $P_{H'}$ and $U_{H'}$  needed to estimate $\bra{E_0} K \ket{E_0}$ within error $\epsilon$ and probability of failure $\delta$ is at most
    $$ \widetilde{O}\left(\frac{|b|^4 \kappa^2 \log(1/\delta)}{\epsilon}\left(\sqrt{\left( \frac{1}{2} + \frac{\langle K_{\max}\rangle }{2|b|^2\kappa}\right)\left(\frac{1}{2} - \frac{\langle K_{\min}\rangle }{2|b|^2\kappa} \right)}  +\sqrt{\frac{\epsilon}{|b|^4 \kappa^2}} \right) \right)
    $$
\end{corollary}
\begin{proof}
    From Theorem~\ref{thm:blockEnc} the number of queries to the prepare and select circuits $P_V,P_{H'}$ and $U_{V},U_{H'}$ needed to $\epsilon_1$-block encode $K$ is
    \begin{equation}
        N_{\rm queries}(U_k) \in O(|b|^2 \kappa\log(\kappa |b|^2/\epsilon_1)).
    \end{equation}
    From Lemma~\ref{lem:KBE} the cost number of queries to a perfect block-encoding of $K$ needed to compute $\bra{E_0} K\ket{E_0}$ within error $\epsilon_2$ is
    \begin{equation}
        O\left(\frac{\alpha\log(1/\delta)}{\epsilon_2}\left(\sqrt{\left( \frac{1}{2} + \frac{\langle K_{\max}\rangle }{2\alpha}\right)\left(\frac{1}{2} - \frac{\langle K_{\min}\rangle }{2\alpha} \right)}  +\sqrt{\epsilon_2/\alpha} \right) \right)
    \end{equation}
    From Box 4.1 from~\cite{nielsen2010quantum} the approximation error grows linearly with the number of times that a faulty unitary is queried.  Thus the error from querying $K$ a total of $N_{\rm queries}(U_K)$ times is at most
    \begin{equation}
        N_{\rm queries}(U_k) \epsilon_2 \in O(\epsilon_2|b|^2 \kappa\log(\kappa |b|^2/\epsilon_1)).
    \end{equation}
    Thus it suffices to choose \begin{equation}
        \epsilon_2=\epsilon_1 = \frac{-\epsilon}{|b|^2 \kappa {\rm LambertW}(-\epsilon/|b|^2\kappa)} \in \widetilde{O}\left(\frac{\epsilon}{|b|^2\kappa} \right).
    \end{equation}
    As the overall number of queries is the product between the queries per invocation of the block-encoding of $K$ and the number of block encodings of $K$ required the total number of queries to our fundamental oracles are in
    \begin{equation}
\widetilde{O}\left(\frac{\alpha|b|^2 \kappa \log(1/\delta)}{\epsilon}\left(\sqrt{\left( \frac{1}{2} + \frac{\langle K_{\max}\rangle }{2\alpha}\right)\left(\frac{1}{2} - \frac{\langle K_{\min}\rangle }{2\alpha} \right)}  +\sqrt{\frac{\epsilon}{\alpha|b|^2 \kappa}} \right) \right)
    \end{equation}
    in order to achieve error $\epsilon$ in the expectation.  This simplifies to the cited result for $\alpha\in \Theta(|b|^2\kappa)$. 
    The $\delta$ scaling is unaffected as no other probabilistic elements are introduced in this process.
\end{proof}




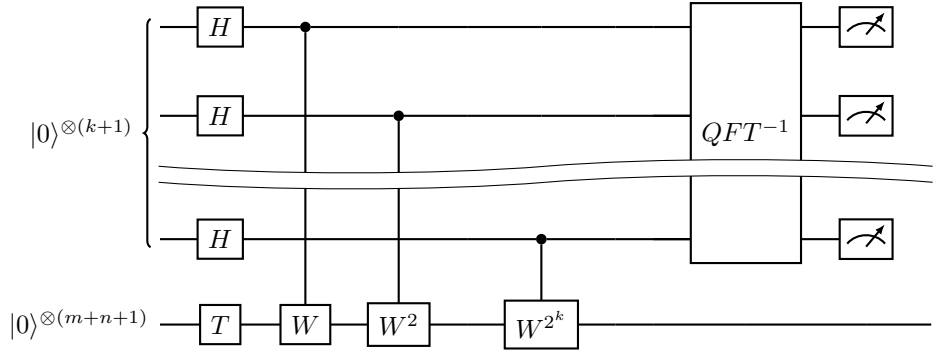
\begin{figure}[t]
    \centering
    \begin{quantikz}
\lstick[4]{$|0\rangle^{\otimes (k+1)}$} & \gate{H} & \ctrl{4} & \qw & \qw & \qw & \qw & \qw  & \gate[4]{QFT^{-1}}  & \meter{} \\
 & \gate{H} & \qw & \ctrl{3} & \qw & \qw & \qw & \qw  & \qw & \meter{} \\
\wave&&&&&&&&&&\\
& \gate{H} & \qw  & \qw & \qw & \ctrl{1} & \qw & \qw & \qw  & \meter{} \\
\lstick{$|0\rangle^{\otimes (m+n+1)}$} & \gate{T} & \gate{W} & \gate{W^2} & \qw & \gate{W^{2^k}} & \qw & \qw  & \qw & \qw & \qw
\end{quantikz}

    \caption{Quantum circuit representation for amplitude estimation for $M=2^{k+1}$}
    \label{fig:QAE}
\end{figure}

\begin{corollary}\label{cor:query}
    Under the assumptions of Corollary~\ref{cor:FreeE} we have that the quantity $\theta_i$, given that $\bra{E_0} K\ket{E_0} \ge k_{\min}$, for any $i$ can be estimated within error $\epsilon$ using a number of queries that is in
    $$
 \widetilde{O}\left(\frac{|b|^4 \kappa^2 \log(1/\delta)}{\epsilon\sqrt{\mu_{\min}}}\left(\sqrt{\left( \frac{1}{2} + \frac{\langle K_{\max}\rangle }{2|b|^2\kappa}\right)\left(\frac{1}{2} - \frac{\langle K_{\min}\rangle }{2|b|^2\kappa} \right)}  +\sqrt{\frac{\epsilon\sqrt{\mu_{\min}}}{|b|^4 \kappa^2}} \right) \right)
    $$
\end{corollary}
\begin{proof}
    By definition,
    \begin{equation}
        \theta_i = \frac{\hbar}{k_b}\sqrt{k_i/\mu_i}
    \end{equation}
    We then have that our estimate of $\theta_i$, $\tilde{\theta}_i$, is expressed as
    \begin{equation}
        \tilde\theta_i = \frac{\hbar}{k_b}\sqrt{\tilde{k}_i/\mu_i}
    \end{equation}
    It then follows that if $\epsilon<k_{\min}/2$ then
    \begin{equation}
        |\theta_i - \tilde\theta_i| \le \frac{\hbar\epsilon}{k_b \sqrt{\mu_i}}.  
    \end{equation}
    Thus the result follows by substitution of $\mu_i\ge \mu_{\min}$ and the assumption that $\hbar$ and $k_b$ are in $O(1)$.
\end{proof}

\begin{theorem}
    Under the assumptions of Corollary~\ref{cor:query}, we have that the free energy $S_{\rm vib}$ can be estimated within error $\epsilon \in o(\min_i\theta_i)$, with probability of failure at most $\delta$, using a number of queries to the oracles $U_V, U_{H'}$ and their corresponding prepare oracles that scales as
    $$
    \widetilde{O}\left(\frac{\mathcal{Z}|b|^4 \kappa^2 \log(1/\delta)}{\epsilon T\sqrt{\mu_{\min}}}\left(\sqrt{\left( \frac{1}{2} + \frac{\langle K_{\max}\rangle }{2|b|^2\kappa}\right)\left(\frac{1}{2} - \frac{\langle K_{\min}\rangle }{2|b|^2\kappa} \right)}  +\sqrt{\frac{\epsilon T\sqrt{\mu_{\min}}}{\mathcal{Z}|b|^4 \kappa^2}} \right) \right)
    $$
    where $\mathcal{Z}:=\sum_i ({e^{\theta_{i}/2T} - 1})^{-1}$
\end{theorem}
\begin{proof}
The expression for the vibrational energy is written as
    \begin{equation}
S_{\text{vib}} = k_b \sum_{i} \left( \frac{\theta_i/T}{e^{\theta_i/T} - 1} - \ln \left( 1 - e^{-\theta_i/T} \right) \right)
\end{equation}
Thus as there are $N$ summands, the it suffices to compute each within error $\epsilon/N$ in order to ensure that the total error is at most $\epsilon$ by the triangle inequality.  First, let us assume that we can compute $\theta_i$ within error $\epsilon'$. 
 We then have that if $\theta_{\min,i}$ denotes the minimum value of the true $\theta_i$ and its estimate  then
\begin{equation}
    |\ln(1-e^{-\theta_i/T}) - \ln(1-e^{-\tilde\theta_i/T})| = \left|\int_{\tilde{\theta}_i/T}^{{\theta}_i/T} \frac{e^{-s}}{1-e^{-s}} ds  \right|\le \frac{\epsilon'}{(e^{\theta_{\min,i}/T}-1)T}
\end{equation}
Next we have that
\begin{equation}
    \left|\frac{\theta_i/T}{e^{\theta_i/T} - 1} - \frac{\tilde\theta_i/T}{e^{\tilde\theta_i/T} - 1}\right| \le \frac{\epsilon'}{(e^{\theta_{\min,i}/T}-1)T} + \frac{\tilde{\theta}_i}{T}\left|\frac{1}{e^{\theta_i/T} - 1} -\frac{1}{e^{\tilde\theta_i/T} - 1} \right| \le \frac{5\epsilon'}{2T({e^{\theta_{\min,i}/T} - 1})}
\end{equation}
if $\epsilon'\le \theta_i/2$ using the same bounding strategy used above and the bound that $e^x \ge 1+x$.  We then have that 
\begin{equation}
    \frac{5\epsilon'}{2T({e^{\theta_{\min,i}/T} - 1})} \le \frac{5\epsilon'}{2T({e^{\theta_{i}/2T} - 1})}
\end{equation}
We then have that the error in the approximate vibrational energy calculation is
\begin{equation}
    |\tilde{S}_{\rm vib}-{S}_{\rm vib}| \le k_b \sum_i \frac{5\epsilon'}{2T({e^{\theta_{i}/2T} - 1})}
\end{equation}
Thus if we wish this error to be at most $\epsilon$ then it suffices to take 
\begin{equation}
    \epsilon' \in O\left(\epsilon T/ \sum_i ({e^{\theta_{i}/2T} - 1})^{-1}\right):= O\left(\frac{\epsilon T}{\mathcal{Z}} \right)
\end{equation}
The number of queries needed in order to reach this level of error is given by Corollary~\ref{cor:query} to be
\begin{align}
     \widetilde{O}\left(\frac{\mathcal{Z}|b|^4 \kappa^2 \log(1/\delta)}{\epsilon T\sqrt{\mu_{\min}}}\left(\sqrt{\left( \frac{1}{2} + \frac{\langle K_{\max}\rangle }{2|b|^2\kappa}\right)\left(\frac{1}{2} - \frac{\langle K_{\min}\rangle }{2|b|^2\kappa} \right)}  +\sqrt{\frac{\epsilon T\sqrt{\mu_{\min}}}{\mathcal{Z}|b|^4 \kappa^2}} \right) \right)
\end{align}
\end{proof}

At a high level, this result shows that we can efficiently approximate the free energy difference  if
\begin{enumerate}
    \item The inversion problem is well conditioned: $\kappa\in O({\rm poly}(n))$
    \item The normalizing function $\mathcal{Z}$ is small, meaning that there are few $\theta_i$ that are near $0$: $\mathcal{Z} \in O({\rm poly}(n))$.
    \item The masses used for the springs are bounded away from zero $\mu_{\min} \in \Theta(1)$
\end{enumerate}

\section{Computational Hardness of Second Derivative}
The above discussion shows that there exists, under appropriate assumptions about the condition number of the system, a quantum algorithm that can estimate the second derivative of the energy exists and can run in polynomial time.  The question remaining is whether it is expected that a classical algorithm exists that could replicate this polynomial scaling.  We show below that, in full generality, that this is unlikely because the problem of estimating the second derivative of the groundstate energy of a Hamiltonian is (promise) \BQP-complete.

\begin{theorem}
    The problem of evaluating the second derivative of an energy eigenstate within error that is poly-logarithmic in the Hilbert space dimension for a well conditioned Hamiltonian that has a LCU decomposition as $\sum_j \alpha_j U_j$ for $\alpha=\sum_{j}|\alpha_j|$  polylogarithmic in Hilbert space dimension and $U_j$ unitaries drawn from the Pauli group on $n$ qubits is \BQP-complete.
\end{theorem}
\begin{proof}
    Let us begin with showing Hardness.  Assume that we have an algorithm that is capable of approximating  the second-derivative of the groundstate energy of an arbitrary Hamiltonian within error $O({\rm poly}(1/n))$.  Our construction for this involves perturbing the Feynman Kitaev Circuit to Hamiltonian construction~\cite{feynman1986quantum,aharonov2002quantum,breuckmann2014space}.  The Feynman Kitaev Circuit to Hamiltonian construction involves taking a quantum circuit of the form $\ket{0}^{\otimes N}\rightarrow U_L\cdots U_{1} \ket{0}^{\otimes n}$ for a set of unitary gates in $\{U_j\}$ in $U(2^n)$ where $U_0:=I$ is the identity gate and $L\in O({\rm poly}(n))$.  As $U_j$ can be taken to be one and two-qubit gates, each term can be decomposed into a um of $O(1)$ Pauli operators, which meets the assumptions required.  Evaluating the probability that the first qubit is zero and under the promise that the probability is either greater than $2/3$ or less than $1/3$ is a promise-\BQP~complete problem.  The circuit to Hamiltonian construction then provides a Hamiltonian whose groundstate is promised to be of the form
    \begin{equation}
        \ket{\psi_0}:= \frac{1}{\sqrt{L+1}}\sum_{j=0}^{L} \ket{j}\prod_{q=0}^j U_j \ket{0}. 
    \end{equation}
    Thus the groundstate expectation value of the projector $\ket{0}\!\bra{0}\otimes \ket{L}\!\bra{L}$ will occur with probability at least $2/3L$ for a yes-instance of the problem and will occur with probability at most $1/3L$ for a no-instance.  Thus this decision can be made by estimating the expectation value within error $\epsilon = O(1/L)$, which by definition is polynomial.  More specifically, the Hamiltonian is of the form~\cite{breuckmann2014space}
    \begin{equation}
        H_{FK}:= -\sum_{t=1}^L U_t \otimes \ket{t}\!\bra{t-1} - U_t^\dagger \otimes \ket{t-1}\!\bra{t} + \ket{t}\!\bra{t} + \ket{t-1}\!\bra{t-1}
    \end{equation}
    
    For our argument, we need to provide a Hamiltonian such that the second derivative can be used to solve the above decision problem.  Before doing so it is useful to review the remainder of the spectrum of $H_{FK}$.  The relevant eigenstates (that are connected to the initial state $\ket{0}$) are~\cite{breuckmann2014space}
    \begin{align}
        \ket{\psi_k}&:= \frac{1}{\sqrt{L+1}}\sum_{j=0}^{L} e^{i2\pi jk/{L+1}}\ket{j}\prod_{q=0}^j U_j \ket{0}\nonumber\\
        H_{FK}\ket{\psi_k} &= E_k\ket{\psi_k}=2(1-\cos(\pi k/(L+1)))\ket{\psi_k}.\label{eq:eigenvalues}
    \end{align}
    This implies that $E_k-E_0 \in O(1/L)$ for all $k$.  Also we define a phase rotation
    \begin{equation}
        R_{+} := \sum_j e^{2\pi j/(L+1)} \ket{j}\!\bra{j}\otimes I
    \end{equation}
    We then have that for any $k$
    \begin{equation}
        R_+ \ket{\psi_k}= \frac{1}{\sqrt{L+1}}\sum_{j=0}^{L} e^{i2\pi j(k+1)/{L+1}}\ket{j}\prod_{q=0}^j U_j \ket{0} =\ket{\psi_{k+1}}.\label{eq:psikDisp}
    \end{equation}

    Consider the following operator valued function $H: [0,1]\rightarrow \mathbb{C}^{2^{n+\log(L+1)}\times {2^{n+\log(L+1)}}}$
    \begin{equation}
        H(s) = H_{FK} + s ( R_+(I
        \otimes (\ket{0}\!\bra{0}\otimes I))  + (I \otimes (\ket{0}\!\bra{0}\otimes I))R_+^\dagger).
    \end{equation}
    The second derivative of the ground state energy with respect to $s$ is given by~\eqref{eq:2deriv} as
    \begin{equation}
        \partial_s^2 E_0(s)|_{s=0} = 2 \sum_{k\ne 0} \frac{|\bra{\psi_k} (R_+(I\otimes \ket{0}\!\bra{0}\otimes I) +I\otimes (\ket{0}\!\bra{0}\otimes I))R_+^\dagger\ket{\psi_0}|^2}{E_0-E_k} 
    \end{equation}
    Using the definition of $R_+$ we can see that this operation translates the initial eigenstates to orthogonal states using~\eqref{eq:psikDisp}. From this we see that
    \begin{align}
        \partial_s^2 E_0(s)|_{s=0} &= 2 \sum_{k\ne 0} \frac{|\sum_j (e^{2\pi i(k-1)j/L+1} + e^{-2\pi i(k-(L-1))j/L+1})\bra{0}  (\prod_{q=j}^0 U_q^\dagger)( \ket{0}\!\bra{0}\otimes I) \prod_{q=0}^j U_q \ket{0}|^2}{(L+1)^2(E_0-E_k)} 
    \end{align}
    Assume that $U_j = I $ for all $j\ge \sqrt{L}$.  The sum of the first $\sqrt{L}$ terms is by the triangle inequality $O(1/L|E_0-E_k|)$ for each $k$.  As the last steps are identity we have that the latter part of the sum for any $k\not\in \{1,L-1\}$ is
    \begin{align}
        &2 \frac{|\sum_{j\ge \sqrt{L}} (e^{2\pi i(k-1)j/L+1} + e^{-2\pi i(k-(L-1))j/L+1})\bra{0}  (\prod_{q=\sqrt{L}}^0 U_q^\dagger)(\ket{0}\!\bra{0}\otimes I) \prod_{q=0}^{\sqrt{L}} U_q \ket{0}|^2}{(L+1)^2(E_0-E_k)} \nonumber\\
        &\quad\in O\left(\frac{1}{L|E_0-E_k|} \right).
    \end{align}
    We have from~\eqref{eq:eigenvalues} that $|E_0-E_k| \in \Omega(1/L^2)$.
    Thus as $L\rightarrow \infty$ the contribution of all terms for $k\not\in\{1,L-1\}$ is in $O(L^{3/2})$.  
    
    For the remaining two terms, the $k=1$ term dominates due to the gap condition.  This results in

    \begin{align}
        \partial_s^2 E_0(s)|_{s=0}= \frac{(L+1-\sqrt{L})^2|\bra{0}  (\prod_{q=j}^0 U_q^\dagger)( \ket{0}\!\bra{0}\otimes I) \prod_{q=0}^j U_q \ket{0}|^2}{(L+1)^2(E_0-E_1)}+ O(L^{3/2}) 
    \end{align}
    Thus
    \begin{equation}
         \frac{(L+1)^2(E_0-E_1)\partial_s^2 E_0(s)|_{s=0}}{(L+1-\sqrt{L})^2}= |\bra{0}  (\prod_{q=j}^0 U_q^\dagger)( \ket{0}\!\bra{0}\otimes I) \prod_{q=0}^j U_q \ket{0}|^2+ O(L^{-1/2}) 
    \end{equation}
    Thus if we can compute the second derivative within error $O(L^{3/2})$ then the error is $O(L^{-1/2})$ which vanishes in the limit as $L\rightarrow \infty.$  The righthand side of the above equation is the square of the probability of measuring a qubit to be zero in an arbitrary quantum circuit of length $\sqrt{L}$ drawn from a universal gate set.  As we assumed that any case that evaluates to $0$ has a probability greater than $2/3$ and otherwise probability less than $1/3$, derivative estimation within relative error $O(1/{\rm poly}(L))$ will solve a promise \BQP-complete problem.  Since $L\in O({\rm poly}(n)$ by assumption, the problem is (promise) $\BQP$-hard.  

    First note that $E_0-E_k$ is $O(1/L^2)$ here which means that if $L\in O({\rm poly}(n))$ then $\kappa$ in $O({\rm poly}(n)$.  The block encoding constant $\alpha$ is similarly $O(L)$ for $H(s)$.  Therefore we have from the above theorem that an efficient simulation can be performed for this operator and in turn the problem can be solved in $\BQP$.  Thus the problem is $\BQP$-complete.
\end{proof}

We see from the above discussion that computing the second derivative of the energy is $\BQP$-complete; however, this does not necessarily mean that this approach to compute the second derivative is computationally challenging for specific hard molecular systems.  This is because the construction used in the above proof requires a Hamiltonian that does not naturally occur in chemical systems and as such, the argument does not necessarily suggest that our approach provides exponential advantage for a particular system.  It does, however, strongly suggest that no method can be constructed that can solve the problem of energy derivative estimation that works in a black box setting (i.e. one that does not take advantage of particular properties of the molecules that we wish to simulate).
\section{Discussion}

In this work, we have developed a comprehensive quantum algorithm for estimating the vibrational entropy of a quantum system. Our approach focuses on evaluating the expectation value of the second derivative of the system energy. By employing state preparation, block encoding, and amplitude estimation, we constructed and analyzed the algorithm in detail. Our error and complexity analyses demonstrate that the algorithm is efficient and manageable for small error bounds.

The presented quantum algorithm provides a robust method for estimating vibrational frequency and, consequently, vibrational entropy, which is critical for understanding the thermodynamic properties of molecular systems. By leveraging quantum computing techniques, our approach overcomes the limitations of classical methods that rely heavily on experimental spectroscopy. Our analysis shows that the algorithm is practical for implementation on near-term quantum computers with fewer computational resources.

The primary application of this work lies in the field of quantum chemistry, particularly in the estimation of thermodynamic properties such as Gibbs free energy and entropy. Accurate estimation of vibrational entropy is crucial for various applications, including understanding molecular stability and reactions at different temperatures, designing materials with specific thermal properties, and predicting the behavior of molecules in biological systems. Moreover, the techniques developed here can be extended to other quantum mechanical observables, enhancing the toolkit available for quantum simulations. Another application of estimating vibrational frequencies is to validate the optimal geometry of a molecule. With a slight modification to the Hadamard test, we can estimate the imaginary part of $\langle K \rangle$, which can be used as an objective function for optimizing molecular geometry.

There are several avenues for future research that build on the foundation laid by this work. One promising direction is further optimization of the algorithm to reduce computational resources and enhance error tolerance. This could involve developing new quantum circuits or algorithms that achieve the same goals more efficiently. Another important direction is the experimental implementation of near-term algorithms, such as variational quantum algorithms, on current quantum hardware. This will allow us to evaluate their practical performance and identify areas for improvement based on real-world data. Additionally, extending the methodology to estimate other important quantum mechanical observables, such as magnetic properties and higher-order energy derivatives, is a natural progression of this work. By doing so, we can broaden the applicability of our approach to a wider range of problems in quantum chemistry and beyond.

In conclusion, the quantum algorithm developed in this work offers a promising pathway for the accurate and efficient estimation of vibrational entropy. This has significant implications for various fields within and beyond quantum chemistry. Continued research and development in this area will further enhance our understanding and capabilities in quantum simulations and thermodynamics, potentially leading to new discoveries and innovations across multiple scientific disciplines.

\section*{Acknowledgements}
NW's work on this project was primarily supported by the U.S. Department of Energy, Office of Science, National Quantum Information
Science Research Centers, Co-design Center for Quantum Advantage (C2QA) under contract number DESC0012704 (PNNL FWP 76274). NW’s research is also
supported by PNNL’s Quantum Algorithms and Architecture for Domain Science (QuAADS) Laboratory Directed Research and Development (LDRD) Initiative.

\bibliographystyle{unsrt}
\bibliography{sample}
\appendix

\section{Derivation of the Coefficients $c_n$}\label{app:cn}

The preparation operator \(\text{P}_{+}\) involves coefficients \( c_n \) defined as:
\begin{equation}
c_n = 4 \sum_{m=0}^{\left\lfloor \frac{n-1}{2} \right\rfloor} (-1)^m \left[ \frac{\sum_{i=m+1}^{B} \binom{2B}{B+i}}{2^{2B}} \right] \frac{(4m+2)(n - (2m+1)) \cdot 2^{2m+1}}{6m+3-n}
\end{equation}
with \( B = \kappa^{2}\log^2\left(\frac{2\kappa}{\epsilon}\right) \), \(\kappa = \frac{\|H'\|}{\Delta E_{\text{min}}} < \frac{E_0 + |a|}{\Delta E_{\text{min}}} \), and \(|c| = \sum_{n=0}^N c_n\).

Here, \( N = \lceil \sqrt{B \log\left(\frac{8B}{\epsilon}\right)} \rceil \). We have:
\begin{equation}
|c|^2 \leq 16N^2 = 16B \log\left(\frac{8B}{\epsilon}\right) = \kappa^{2}\log^2\left(\frac{2\kappa}{\epsilon}\right)
\end{equation}

To derive and simplify the sum of the coefficients \( c \), we recognize that summing the coefficients of a polynomial is equivalent to evaluating the polynomial at \( x = 1 \). Therefore, we have:
\begin{equation}
c = g(1)
\end{equation}

Given the function:
\begin{equation}
g(x) = 4 \sum_{n=0}^{N} (-1)^n \left[ \frac{\sum_{i=n+1}^{B} \binom{2B}{B+i}}{2^{2B}} \right] \mathcal{T}_{2n+1}(x),
\end{equation}
we need to evaluate it at \( x = 1 \).

First, recall that the Chebyshev polynomials of the first kind, \( \mathcal{T}_n(x) \), have the property:
\begin{equation}
\mathcal{T}_n(1) = \cos(n \arccos(1)) = \cos(0) = 1 \quad \text{for all } n.
\end{equation}

Substituting \( x = 1 \) into \( g(x) \), we obtain:
\begin{equation}
g(1) = 4 \sum_{n=0}^{N} (-1)^n \left[ \frac{\sum_{i=n+1}^{B} \binom{2B}{B+i}}{2^{2B}} \right] \mathcal{T}_{2n+1}(1).
\end{equation}
Since \( \mathcal{T}_{2n+1}(1) = 1 \), this simplifies to:
\begin{equation}
g(1) = 4 \sum_{n=0}^{N} (-1)^n \left[ \frac{S_n}{2^{2B}} \right],
\end{equation}
where \( S_n = \sum_{i=n+1}^{B} \binom{2B}{B+i} \).

Therefore, the sum of the coefficients \( c \) is given by:
\begin{equation}
c = g(1) = 4 \sum_{n=0}^{N} (-1)^n \left[ \frac{\sum_{i=n+1}^{B} \binom{2B}{B+i}}{2^{2B}} \right] \leq 4N.
\end{equation}

\end{document}